\documentclass[a4paper,reqno]{amsart}

\usepackage{xspace}
\usepackage[textsize=small]{todonotes}

\newcommand*{\jf}[1]{\relax}

 \makeatletter
 \providecommand\@dotsep{5}
 \def\listtodoname{List of Todos}
 \def\listoftodos{\@starttoc{tdo}\listtodoname}
 \makeatother

\usepackage[scr]{rsfso}
\usepackage[sc,osf]{mathpazo}

\usepackage{amssymb}
\usepackage{mathtools}

\newtheorem{thm}{Theorem}

\usepackage[utf8]{inputenc}
\usepackage[T1]{fontenc}
\usepackage[english]{babel}
\usepackage{csquotes}
\usepackage{enumitem}
\usepackage[pdfpagelabels,pdftex]{hyperref}

\hypersetup{%
        linkcolor=black,                % Für Links in der gleichen Seite
        urlcolor=blue,                 % Für Links auf URLs
        breaklinks=true,                % Links dürfen umgebrochen werden
        colorlinks=false,               % Links farbig markieren oder mit Box umfassen
        citebordercolor=0 0 0,          % Farbe für \cite
        filebordercolor=0 0 0,          % Farbe für Fileboxen
        linkbordercolor=0 0 0,
        menubordercolor=0 0 0,          %
        urlbordercolor=0 0 0,           % Farbe für Urlboxen
        pdfhighlight=/I,                %
        pdfborder=0 0 0,                % keine Box um die Links!
}

\numberwithin{equation}{section}
\usepackage[%backend=bibtex,
            style=phys,
            sorting=nyt,
            articletitle=true,
            biblabel=brackets,
            chaptertitle=true,
            pageranges=true,
            hyperref=true,
            maxcitenames=3,
            maxbibnames=4,
            url=true,
            autocite=superscript
            ] {biblatex}

\newcommand{\RR}{{\mathbb {R}}}
\newcommand{\CC}{{\mathbb{C}}}
\newcommand{\mb}{\overline{m}}
\newcommand{\Lie}[2]{\mathscr L_{#1}{#2}}

\renewcommand*{\i}{\mathrm{i}}
\newcommand*{\dd}{\mathrm{d}}
\newcommand*{\e}{\mathrm{e}}
\newcommand*{\del}{\partial}

\ifx\thorn\undefined
\newcommand*{\thorn}{\mbox{\th}}
\else
\renewcommand*{\thorn}{\mbox{\th}}
\fi
\newcommand*{\thorp}{\thorn'}
\newcommand*{\etp} {\eth'}
\newcommand*{\scri}{\mathscr{I}}

\newcommand*{\cC}{\mathcal{C}}
\newcommand*{\cT}{\mathcal{T}}
\newcommand*{\cS}{\mathcal{S}}
\newcommand*{\cO}{\mathcal{O}}
\newcommand*{\cP}{\mathcal{P}}
\newcommand*{\cQ}{\mathcal{Q}}

\newcommand*{\sB}{\mathscr{B}}

\newcommand*{\sS}{\mathscr{S}}
\newcommand*{\sT}{\mathscr{T}}

\newcommand*{\news}{\mathscr{N}}

\newcommand*{\ethc}[1][]{{\eth^{#1}_c}}
\newcommand*{\etbc}[1][]{{\overline{\eth}{}^{#1}_c}}
\newcommand*{\thorpc}{{\thorp_c}}

\newcommand*{\mbb}{\mathbf{m}}
\newcommand*{\lbb}{\mathbf{l}}
\newcommand*{\mbbar}{\overline{\mathbf{m}}}
\newcommand*{\mub}{\pmb{\mu}}

\newcommand*{\Nb}{\overline{N}}
\newcommand*{\Bb}{\overline{\B}}

\newcommand*{\zb}{{\bar{z}}}

\newcommand*{\vb}{{\bar{v}}}

\newcommand*{\Pb}{\overline{\cP}}
\newcommand*{\Qb}{\overline{\cQ}}

\newcommand*{\rhob}{{\bar\rho}}
\newcommand*{\sigmab}{{\bar\sigma}}

\newcommand*{\taub}{{\bar\tau}}

\newcommand*{\psib}{{\bar\psi}}
\newcommand*{\xib}{{\bar\xi}}

\newcommand*{\B}{R}
\newcommand*{\fM}{\mathfrak{m}}

\DeclareMathOperator{\ad}{ad}

\begin{document}
\title{A new look at the Bondi-Sachs energy-momentum}
\author[J. Frauendiener]{J\"org Frauendiener}
\address{Department of Mathematics and Statistics, University of Otago, PO Box 56, Dunedin 9054, New Zealand}
\email{joergf@maths.otago.ac.nz}
\author[C. Stevens]{Chris Stevens}
\address{Department of Mathematics and Statistics, University of Otago, PO Box 56, Dunedin 9054, New Zealand}
\email{cstevens@maths.otago.ac.nz}

\begin{abstract}
How does one compute the Bondi mass on an arbitrary cut of null infinity $\scri$ when it is not presented in a Bondi system? What then is the correct definition of the mass aspect? How does one normalise an asymptotic translation computed on a cut which is not equipped with the unit-sphere metric? These are questions which need to be answered if one wants to calculate the Bondi-Sachs energy-momentum for a space-time which has been determined numerically. Under such conditions there is not much control over the presentation of $\scri$ so that most of the available formulations of the Bondi energy-momentum simply do not apply. The purpose of this article is to provide the necessary background for a manifestly conformally invariant and gauge independent formulation of the Bondi energy-momentum. To this end we introduce a conformally invariant version of the GHP formalism to rephrase all the well-known formulae. This leads us to natural definitions for the space of asymptotic translations with its Lorentzian metric, for the Bondi news and the mass-aspect. A major role in these developments is played by the ``co-curvature'', a naturally appearing quantity closely related to the Gauß curvature on a cut of~$\scri$.
\end{abstract}
\thanks{Supported by the Marsden Fund Council from Government funding, managed by Royal Society Te Apārangi. JF is grateful to L. Szabados for several discussions on this topic.}
\maketitle

\section{Introduction}
\label{sec:introduction}

One of the main achievements in the theoretical development of Einstein's theory of gravity was the discovery that a gravitating system may lose energy through the emission of gravitational waves. The famous ``mass-loss formula'' was derived by H.~Bondi and his group in 1962~\cite{Bondi:1962} in the axisymmetric case and shortly after by R.~Sachs in the general case~\cite{Sachs:1962a}. The main consequence of the mass-loss formula is the proof that --- provided that Einstein's theory is correct --- gravitational waves must exist, nowadays an experimentally well established fact~\cite{Abbott:2016,Abbott:2016a,Abbott:2017,Abbott:2017a}.

The mass-loss formula involves two quantities, the energy flux due to gravitational waves and the total mass of the system in question. By their very nature these quantities are defined at infinity where the global properties of a space-time reside. The energy flux measures the intensity of the gravitational radiation carried by the waves in each direction away from the system at each instant of time. The total mass of the system depends on time and is computed by a surface integral of the mass aspect over the sphere of all outgoing directions at an instant of retarded time.

Obviously, the Bondi mass and the gravitational flux are crucially important quantities describing --- at least in part --- a gravitationally active system. However, it is not straightforward to ``measure'' these quantities. Clearly, the description above implies that the Bondi mass is a global quantity which resides at infinity. But things are even more complicated. The Bondi mass is computed from the components of the Bondi-Sachs energy-momentum, a 4-covector defined on a cut of null infinity. Physically speaking, a cut is the idealised spherical surface of instants in retarded time when observers distributed in all directions at infinite distance from the system  measure the gravitational wave signal. This information provides the so called \emph{mass-aspect} on the cut and the Bondi-Sachs 4-momentum is obtained by integrating the mass-aspect against four functions defined on the cut which can be interpreted in some sense as translations. This is in line with the fact that energy and momentum are associated in physics with translational symmetries~\cite{Ashtekar:1981}.

In the mathematical treatment of the situation there is some gauge freedom in the description of null infinity and in all discussions of the Bondi energy-momentum this freedom was used to simplify the description as much as possible. The freedom in the description consists of essentially three types, the choice of coordinates, the choice of a frame and the choice of a conformal factor, when the approach is based on Penrose's conformal treatment of asymptotically flat space-times~\cite{Penrose:1965}. When these simplifying choices are made then the resulting formulae for the Bondi energy-momentum are deceptively simple. However, when one approaches null infinity in a way which is not in line with these simplifying assumptions then the simple formulae are no longer valid. But this is exactly the situation which one is facing when space-times are computed numerically using codes which are capable of reaching null infinity in finite time such as the codes based on the characteristic formulations of the Einstein equations \cite{Bishop:1997,Winicour:2001} or the conformal field equations~\cite{Hubner:1996a,Hubner:2001a,Frauendiener:1998,Frauendiener:1998a,Frauendiener:2002c,Frauendiener:2004}. In these cases, the gauge is dictated by the  formulation of the equations which makes the numerical treatment as well-behaved as possible and this is, in general, not the same as the choices needed for simplifying the treatment of null infinity.

In this paper we give a prescription for the determination of the Bondi energy-momentum in a general gauge. This is an important task not only because of the physical relevance of the Bondi mass but also because the Bondi-Sachs mass-loss formula is a very good test for the validity of a numerical code.

The outline of the paper is as follows. In sec.~\ref{sec:asympt-flatn-struct} we collect the necessary facts about null infinity following the standard sources such as e.g.~\cite{Penrose:1986,Frauendiener:2004a,Geroch:1977,ValienteKroon:2016}. Sec.~\ref{sec:conf-ghp-form} is devoted to the introduction of a manifestly conformally invariant formalism based on conformal densities. This formalism is an extension of the more familiar GHP formalism~\cite{Geroch:1973b,Penrose:1984a}. It is used in sec.~\ref{sec:cut-systems-null} and \ref{sec:bms-algebra} to study the structure of cut systems of $\scri$ and of the BMS algebra and, in particular, to derive and characterize the ideal of asymptotic translations in sec.~\ref{sec:lorentz-metric-space}. Finally, in sec.~\ref{sec:mass-aspect} we discuss the Bondi energy-momentum and prove (again) the mass-loss formula for gravitational radiation in asymptotically flat space-times in sec.~\ref{sec:mass-loss-formula}. We finish the paper with a brief description of how one would use the formulae obtained earlier to compute the Bondi energy-momentum in a space-time for which null infinity is not presented in a Bondi gauge.

We use the same conventions and notation as~\cite{Penrose:1986} throughout.

\section{Asymptotic flatness and the structure of null infinity}
\label{sec:asympt-flatn-struct}

In this section we will briefly outline the specific assumptions that are made in the definition of the Bondi energy-momentum. We describe this not from a point of view within the physical space-time but from within a conformally related space-time. We are interested here in an asymptotically flat space-time $(\widetilde{M},\tilde{g}_{ab})$. By definition, this means that we may regard $\widetilde{M}$ as embedded into a larger `unphysical' space-time $(M,g_{ab})$ where the metrics are related on $\widetilde{M}$ by
\begin{equation}
  \label{eq:1}
  g_{ab} = \Omega^2\tilde{g}_{ab}
\end{equation}
for a function $\Omega:M\to\RR$ with $\Omega(x)>0$ if and only if  $x\in\widetilde{M}$. We denote the zero-set of $\Omega$ by $\scri$, usually called \emph{null infinity},  and assume that it is a regular submanifold of $M$ so that $\scri=\{x \in M : \Omega(x)=0, \dd\Omega(x) \ne0\}$. Notice, that the conformal factor $\Omega$ is not unique, any function $\widehat\Omega=\Omega\Theta$ with $\Theta>0$ on $\widetilde{M}\cup\scri$ together with the metric $\hat{g}_{ab}=\Theta^2 g_{ab}$ satisfies the same conditions. We assume that $\scri$ has two connected components $\scri=\scri^+ \cup \scri^-$ each with the topology $S^2\times \RR$, see~\cite{Geroch:1977,Penrose:1965,Frauendiener:2004} for more details. The sets $\scri^\pm$ are called \emph{future} and \emph{past null infinity}. In what follows we will implicitly assume that $\scri$ refers to $\scri^+$. Similar considerations hold for $\scri^-$.

The curvature tensors of the two metrics differ by terms containing derivatives of the conformal factor $\Omega$. The various pieces of the Riemann tensors, i.e., the Weyl tensor $C_{abc}{}^d$, the tensor $\Phi_{ab}=-\frac12 (R_{ab}- \frac14 R g_{ab})$, defined in terms of the Ricci tensor $R_{ab}$, and the curvature scalar $\Lambda=\frac1{24}R$, are related on $\widetilde{M}$ according to the formulae 
\begin{equation}
  \label{eq:2}
  \begin{aligned}
    \widetilde{C}_{abc}{}^d &= C_{abc}{}^d \\
    \Omega\widetilde{\Phi}_{ab} &= \Omega\Phi_{ab}  + \nabla_a\nabla_b\Omega - \frac14 g_{ab} \Box\Omega,\\
    \tilde\Lambda &= \Omega^2\Lambda - \frac14 \Omega \Box \Omega + \frac12 \nabla_a\Omega \nabla^a\Omega.
\end{aligned}
\end{equation}

We assume that the vacuum Einstein equations hold in $\widetilde{M}$ near $\scri$. Thus, the physical Einstein tensor $\widetilde{G}_{ab}$ vanishes and so does the physical Ricci tensor $\widetilde{R}_{ab}$. Hence, the equations
\begin{equation}
  \label{eq:3}
  \begin{aligned}
    0 &= \Omega\Phi_{ab}  + \nabla_a\nabla_b\Omega - \frac14 g_{ab} \Box\Omega,\\
    0 &= \Omega^2\Lambda - \frac14 \Omega \Box \Omega + \frac12 \nabla_a\Omega \nabla^a\Omega
\end{aligned}
\end{equation}
hold on $\widetilde{M}$. Since all geometric quantities are smooth these equations extend smoothly to $\scri$ and we can write them in the form
\begin{align}
    \nabla_a\nabla_b\Omega - \frac14 g_{ab} \Box\Omega &= \cO(\Omega),  \label{eq:4}\\
    \nabla_a\Omega \nabla^a\Omega &= \cO(\Omega).\label{eq:5}
\end{align}
The first of these equations is termed the \emph{asymptotic Einstein condition} in~\cite{Penrose:1986} while the second equation shows that $\scri$ is a regular null hyper-surface. A fundamental consequence of this construction is that the Weyl tensor vanishes on $\scri$. This allows us to introduce the rescaled Weyl quantities $\psi_i=\Omega^{-1}\Psi_i$ for $i=0:4$, smooth complex valued functions on $M$, where $\Psi_i$ are components of $C_{abc}{}^d$ with respect to the null tetrad. For their definition, as for the definition of all the spin-coefficients etc we refer to~\cite{Penrose:1984a}.

As the next step we collect all the equations on $\scri$ that are relevant. We write them down with respect to a null tetrad which is chosen as follows. The metric $g_{ab}$ when restricted to $\scri$ is degenerate  i.e., at every point on $\scri$ there is a 1-dimensional subspace of tangent vectors which kill the metric. Let $n^a$ be a non-zero vector in that subspace. We complement it by a complex null vector $m^a$ and an additional real null vector $l^a$ to a null tetrad $(n^a,m^a,\mb^a,l^a)$ for $M$ at every point on $\scri$. The dual (co-vector) basis is $(l_a,-\mb_a,-m_a,n_a)$. In most of what follows we will be concerned only with quantities which are intrinsic to $\scri$. The corresponding basis and dual basis intrinsic to $\scri$ are obtained by dropping the last members of the 4-dimensional bases.

To simplify things and in view of later discussions we now assume the existence of a scalar function $s:\scri\to\RR$ with the property that $D's=n^a\nabla_as\ne0$ everywhere and we assume that the complex null vector $m^a$ has been chosen so that $\delta s = m^a\nabla_a s=0$. Then, the level sets of constant $s$ are ``cuts'' of $\scri$, i.e., 2-dimensional surfaces everywhere transverse to its null generators, i.e., the integral curves to $n^a$.

With this setup we are now in a position to introduce the GHP-formalism~\cite{Geroch:1973b,Penrose:1984a} and, in particular, the operators $\thorp$, $\eth$ and $\etp$ acting in the directions of the basis vectors tangent to $\scri$. 

The gradient of the conformal factor $\Omega$ on $M$ defines a vector field\footnote{Note, that the sign makes the vector future-pointing on $\scri^+$. On $\scri^-$ one would have a different sign.} $N^a:=-g^{ab}\nabla_b\Omega$ which, when restricted to $\scri$ is null, i.e., such that $N^a=An^a$ for some scalar $A$ on $\scri$. When the conformal factor is changed by a rescaling $\Omega \mapsto \Theta \Omega$ then both $g_{ab}$ and $N^a$ are changed, $g_{ab} \mapsto \Theta^2\,g_{ab}$ and $N^a\mapsto \Theta^{-1}N^a$, but the tensor
\begin{equation}
  \Gamma_{ab}{}^{cd} := g_{ab} N^cN^d\label{eq:6}
\end{equation}
remains unchanged. This is the ``universal structure tensor'' as defined by Geroch~\cite{Geroch:1977}. It is closely related to the ``strong conformal geometry'' defined by Penrose~\cite{Penrose:1986} which is essentially the ``square root'' of the tensor. Its relevance becomes clearer when one considers a generator of $\scri$ with tangent vector $N^a$ and associated  parameter $u$ defined by $\dd u(N)=1$. The parameter is called the Bondi time and it measures the retarded time along that generator. Under a change of conformal factor it changes according to $\dd u \mapsto \dd \hat{u} = \Theta \dd u$. At the same time, the length $\dd l(v)=\sqrt{-g_{ab}v^av^b}$ of any vector $v^a$ tangent to a cut through the generator changes by the same factor $\Theta$ so that the ratio $\dd l:\dd u$ is unchanged. This is a remnant of the fact that, in relativity, space and time intervals do not have independent meaning. The spatial and temporal scales are tied together by the constant speed of light. At every point outside of $\scri$ this is encoded in the Lorentzian signature of the metric. But restricted to $\scri$ where the metric is degenerate this fundamental fact of Einstein's theory is still maintained in the form of this ratio and its conformal invariance. 

Having set up the null tetrad we are now in a position to discuss the implications for the spin-coefficients on $\scri$. The first obvious consequences follow from the fact that $\scri$ is a null hypersurface: it is generated by a null geodetic congruence tangent to $n^a$ so that on $\scri$ the equations $\kappa'=0$ and $\bar\rho' = \rho'$ must hold.The fact that we aligned $m^a$ with the 2-surfaces of constant $s$ implies that also $\bar\rho=\rho$ on $\scri$. More information can be gleaned from the asymptotic Einstein condition~\eqref{eq:4} which after inserting $\nabla_a\Omega = -A n_a + \Omega X_a$ can be rewritten as
\[
  \nabla_a A n_b + A \nabla_a n_b - A n_a X_b = g_{ab} X + \Omega X_{ab}
\]
where the fields $X$, $X_a$ and $X_{ab}=X_{ba}$ are some irrelevant smooth fields on $M$. Taking components and evaluating the combinations which are free of any $X$ fields on $\scri$ yields the equations
\[
  \begin{gathered}
    \kappa' = 0, \qquad \sigma'=0,\qquad \bar\rho' = \rho',\\
    \thorp A + \rho' A = 0 , \qquad \eth A = 0.
  \end{gathered}
\]
The next set of equations comes from the curvature equations which relate the curvature to derivatives of the spin-coefficients. Those equations which are intrinsic to $\scri$ are
\begin{align}
  \label{eq:7}
  \thorp \rho' - \rho'^2 &= \Phi_{22}, \\
  \label{eq:8}
  \etp\rho' &= \Phi_{21}, \\
  \label{eq:9}
  \eth\rho - \etp \sigma &= \Phi_{01}, \\
  \label{eq:10}
  \thorp \sigma - \rho' \sigma - \eth\tau + \tau^2 &= -\Phi_{02}, \\
  \label{eq:11}
  \thorp\rho - \rho\rho' + 2 \Lambda &= \etp\tau - \tau\taub.
\end{align}
Note, that a consequence of~\eqref{eq:11} is that
\[
  \etp\tau = \eth \taub.
\]
The commutators between the three operators also contain information about the spin-coefficients. Restricted to $\scri$ they are when acting on a GHP scalar with weights $(p,q)$
\begin{align}
  \thorp\eth - \eth\thorp &= \rho'\eth - \tau \thorp - p (\rho'\tau - \Phi_{12}),\label{eq:12}\\
  \thorp\etp - \etp\thorp &= \rho'\etp - \taub \thorp - q (\rho'\taub - \Phi_{21}),\label{eq:13}\\
  \eth\etp - \etp\eth &=  - (p - q) K,\label{eq:14}
\end{align}
where $K=\Phi_{11} + \Lambda -\rho\rho' = \bar{K}$ is the remnant of the complex curvature on $\scri$, where in fact it is real.

Finally, we need the relevant Bianchi identities. The Bianchi identity for the physical space-time  yields an equation for the rescaled Weyl spinor components $\psi_i$ in terms of the spin-coefficients which looks superficially like a zero-rest-mass field equation for a spin-2 field. The Bianchi identity for the conformal space-time lead to equations for the Ricci components in terms of the Weyl scalars. Explicitly, the relevant equations intrinsic to $\scri$ are
\begin{itemize}[wide]
\item the intrinsic propagation equations along $\scri$ for $\psi_{ABCD}$
\begin{align}
  \thorp \psi_0 - \rho' \psi_0 - \eth \psi_1 + 4 \tau \psi_1 &= 3 \sigma \psi_2,\label{eq:15}\\
  \thorp \psi_1 - 2\rho' \psi_1 - \eth \psi_2 + 3 \tau \psi_2 &= 2 \sigma \psi_3,\label{eq:16}\\
  \thorp \psi_2 - 3\rho' \psi_2 - \eth \psi_3 + 2 \tau \psi_3 &= \sigma \psi_4,\label{eq:17}\\
  \thorp \psi_3 - 4\rho' \psi_3 - \eth \psi_4 +  \tau \psi_4 &= 0\label{eq:18},
\end{align}
\item and the intrinsic propagation equations along $\scri$ for the Ricci components in terms of the $\psi$'s
\begin{align}
  \thorp \Phi_{00} - \etp \Phi_{01} + 2 \thorn \Lambda &= -2 \tau \Phi_{10} - 2 \taub \Phi_{01} + \rho' \Phi_{00} + \sigmab \Phi_{02} + 2 \rho \Phi_{11} + A \psi_2, \label{eq:19}\\
    \thorp \Phi_{01} - \etp \Phi_{02} + 2 \eth\Lambda &= - 2\tau \Phi_{11} -  \taub \Phi_{02} + 2\rho' \Phi_{01} + 2 \rho \Phi_{12} ,\label{eq:20}\\
  \thorp \Phi_{01} - \eth \Phi_{11} + \eth\Lambda &= -2 \tau \Phi_{11} - \taub \Phi_{02} + \rho' \Phi_{01} + \sigma \Phi_{21} + \rho \Phi_{12} + A \psib_3,\label{eq:21}\\
  \thorp \Phi_{11} - \etp \Phi_{12} + \thorp\Lambda &= - \taub \Phi_{12} - \tau \Phi_{21} + 2 \rho' \Phi_{11} + \rho \Phi_{22} ,\label{eq:22}\\
  \thorp \Phi_{20} - \etp \Phi_{21}  &= -2 \taub \Phi_{21} + \rho' \Phi_{20} + \sigmab \Phi_{22} + A \psi_4,\label{eq:23}\\
  \thorp \Phi_{21} - \etp \Phi_{22}  &= - \taub \Phi_{22} + 2 \rho' \Phi_{21} .\label{eq:24}
\end{align}
\end{itemize}
There are no equations for $\thorp\psi_4$ and $\thorp\Phi_{22}$. Strictly speaking,~\eqref{eq:19} is not intrinsic since it contains $\thorn\Lambda$. However, it is a complex equation and its imaginary part is intrinsic. Furthermore, we have  two equations for $\thorp\Phi_{01}$ which could be combined into a ``propagation equation'' and a ``constraint equation''.

Since we have values for all the Ricci components except $\Phi_{00}$ and $\Phi_{11}$ we would expect that all equations which contain only those known components will give us relation between the remaining spin-coefficients or are satisfied identically. For instance, ~\eqref{eq:24} is easily seen to be an identity once we insert the values on $\scri$.

We will come back to these equations after we have developed some more background.

\section{The conformal GHP formalism}
\label{sec:conf-ghp-form}

Since the existence of $\scri$ is due to the conformal compactification of a physical space-time with a conformal factor that is fixed only up to the multiplication with a positive function the entire physical content on $\scri$ must be invariant under conformal rescalings. This property is usually exploited by choosing a conformal gauge, i.e., by fixing the conformal factor so that calculations simplify. We will  proceed here in a different way and maintain the conformal invariance in all our operations. To this end we need to introduce conformally weighted quantities and conformally invariant derivative operators (see~\cite{Penrose:1984a}).

We call $\eta$ a \emph{conformal density of weight} $w$ if it changes under the change $\Omega\mapsto\Omega \Theta$ by $\eta \mapsto \Theta^w \eta$. Since we are using the GHP formalism which partially implements a frame invariance such quantities $\eta$ will in general also  be GHP weighted. We denote the weights of a conformally weighted GHP quantity $\eta$ in the form $[w;p,q]$. In order to incorporate the conformal invariance within the GHP formalism one needs to decide about how the null tetrad transforms under conformal rescaling. Here, the most natural choice is to map
\[
  l^a \mapsto \Theta^{-2}l^a,\qquad 
  m^a \mapsto \Theta^{-1}m^a,\qquad 
  n^a \mapsto n^a.
\]
Even for a conformal density $\eta$ the derivatives $\thorp \eta$ and $\eth \eta$ will not be conformal densities. Instead we find for $\eta$ of weights $[w;p,q]$ that
\[
  \begin{aligned}
  \thorp \eta &\mapsto \Theta^w(\thorp \eta + (w+p+q) D'\theta \, \eta), \\
  \eth \eta &\mapsto \Theta^{w-1}(\eth \eta + (w+q) \delta\theta \, \eta), \\
  \etp \eta &\mapsto \Theta^{w-1}(\etp \eta + (w+p) \delta'\theta \, \eta).
\end{aligned}
\]
where we have defined $\theta = \log\Theta$. In order to obtain conformal densities one needs to eliminate the derivatives of $\theta$. To this end one uses the inhomogeneous transformation of some of the spin-coefficients, in this case of $\tau$ and $\rho'$ which transform according to
\[
  \tau \mapsto \Theta^{-1}(\tau - \delta \theta), \quad \rho' \mapsto \rho' - D'\theta.
\]
Combining these two inhomogeneous behaviours leads to the introduction of the conformally invariant GHP operators acting on a quantity $\eta:[w;p,q]$
\begin{align}
  \thorpc \eta &= \thorp \eta + (w + p + q)\rho' \eta,\label{eq:25}\\
  \ethc \eta &= \eth \eta + (w + q)\tau \eta,\label{eq:26}\\
  \etbc \eta &= \etp \eta + (w + p)\taub \eta.\label{eq:27}
\end{align}
Note, that in the context of $\scri$ there is no need to consider $\thorn_c$ since it is an outward derivative. Similarly, we do not need the conformal $\etp_c$ operator. This would be defined in terms of $\tau'$ which is also an extrinsic quantity, being defined in terms of the parallel transport of $n^a$ along $l^a$ away from $\scri$. Thus, the two derivatives transverse to the generators but tangent to $\scri$ are represented by $\ethc$ and its complex conjugate $\etbc$, \emph{which is not the same as $\etp_c$}.

For $\eta$ a conformal density of weight $w$ the derivatives have weights
\[
  \thorpc\eta : [w;p-1,q-1], \quad
  \ethc\eta : [w-1;p+1,q-1], \quad
  \etbc\eta : [w-1;p-1,q+1].
\]

Next, we need the commutators among these operators. Some calculation yields the rather simple result which holds on $\scri$ (i.e., with $\sigma' = 0 = \kappa' = \rho' - \rhob'$)
\begin{align}
  [\thorpc,\ethc]\eta &= (w + q) \cP \eta, \label{eq:28}\\
  [\thorpc,\etbc]\eta &= (w + p) \Pb \eta, \label{eq:29}\\
  [\ethc,\etbc]\eta &= -(p-q) \cQ \eta. \label{eq:30}
\end{align}
The commutators define the \emph{quantities}
\begin{equation}
  \cP := \thorp\tau - \eth\rho', \qquad \cQ := K - \etp\tau = \Qb,\label{eq:31}
\end{equation}
which are conformal densities $\cP:[-1;0,-2]$, $\Pb:[-1;-2,0]$ and $\cQ:[-2;0,0]$. Recall that on $\scri$ both $K$ and $\etp\tau$ are real. While they are not conformally invariant individually, their combination is. Similarly, the individual terms defining $\cP$ are not conformally invariant but their combination is a conformal density. This can be checked explicitly by going through the individual transformations.

Finally, we evaluate the Jacobi identity for these operators in order to get equations between these commutator quantities. There is only one non-trivial combination to consider, namely
\[
  [\thorpc,[\ethc ,\etbc]] + [\ethc,[\etbc ,\thorpc]] + [\etbc,[\thorpc ,\ethc]] = 0.
\]
Evaluating the three terms acting on an arbitrary conformal density $\eta:[w;p,q]$ yields
\[
(w+q)\left[\etbc\cP + \thorpc\cQ\right] - (w+p) \left[\ethc \Pb + \thorpc\cQ\right] = 0.
\]
Since this holds for all $\eta$, i.e., for all weights, we get the equations
\begin{equation}
\label{eq:32}
  \thorpc\cQ  + \etbc\cP  = 0, \qquad \thorpc\cQ + \ethc \Pb = 0
\end{equation}
with the consequence that $\ethc \Pb = \etbc\cP$.

Returning now to the equations which are implied on $\scri$. It is easily verified that $\sigma$ is a conformal density with weight $[-2;3,-1]$. Also, the factor $A$ is a conformal density with weights $A:[-1;1,1]$ satisfying the equations
\[
  \thorpc A = 0, \qquad \ethc A = 0
\]
on $\scri$.

The components of the rescaled Weyl tensor are all conformal densities with weights
\[
  \psi_k : [k-5;4-2k,0]
\]
and the equations relating them can be written in terms of the conformal derivative operators as
\begin{equation}
  \thorpc \psi_k - \ethc \psi_{k+1} = (3-k) \sigma \psi_{k+2}, \qquad k=0:3.\label{eq:33}
\end{equation}
The Ricci components are not immediately conformal densities and their equations can not be simply rewritten in a conformally invariant form. However, some of them correspond to equations just derived. Consider for instance~\eqref{eq:22}. Inserting the values for the Ricci components from (\ref{eq:7}--\ref{eq:10}) $\Phi$'s yields
\[
  \begin{multlined}
    \thorp \Phi_{11} + \thorp\Lambda - 2 \rho' \Phi_{11} = \etp \eth\rho' - \taub \eth\rho' - \tau \etp\rho'  + \rho \thorp\rho' - \rho\rho'^2 \\
    = \etp \eth\rho' - \taub \eth\rho' - \tau \etp\rho'  + \thorp(\rho \rho') -\rho'\etp\tau  + \rho'\tau\taub + 2 \rho'\Lambda - 2 \rho\rho'^2
\end{multlined}
\]
which can be rewritten in the form
\[
  \thorp (\Phi_{11} + \Lambda -\rho\rho') - 2 \rho' (\Phi_{11} + \Lambda - \rho\rho') = \etp \eth\rho' - \taub \eth\rho' - \tau \etp\rho'   - \rho'\etp\tau  + \rho'\tau\taub.
\]
Adding the terms $\thorp\etp\tau - 2\rho'\etp\tau$ on both sides and using the $[\thorp,\etp]$ commutator yields the second equation in~\eqref{eq:32}.

\section{Cut systems of null infinity}
\label{sec:cut-systems-null}

It is clear from what was discussed so far that the only non-trivial GHP spin-coefficients on $\scri$ are $\rho$, $\sigma$, $\rho'$ and $\tau$. The first two give information which is extrinsic to $\scri$ in the sense that they tell us about the shape of the outgoing null hypersurfaces intersecting $\scri$ in the cuts of constant $s$. The role of $\sigma$ is closely tied to the gravitational radiation which arrives at infinity. While $\rho$ and $\sigma$ describe dynamical properties, the other two spin-coefficients, $\rho'$ and $\tau$ are intrinsic to $\scri$ and describe its non-dynamical, kinematical, structure. Clearly, $\rho'$ determines the conformal gauge in the sense that given its value on $\scri$ one can recreate the conformal factor $\Theta$ from its value at a single cross-section by solving the equation $\thorpc\Theta=0$. But, what is the meaning of $\tau$?

Recall our assumption that there exists a function $s$ on $\scri$ such that its level surfaces are regular cross-sections of $\scri$ and that the frame is adapted to the cuts so that $\delta s = 0$. We take $s$ to be a scalar, i.e., having weights $[0;0,0]$ and consider the commutator
\[
  [\thorp,\eth]s =  \thorp\eth s - \eth\thorp s = \rho' \eth s - \tau \thorp s
\]
which implies that $\tau = \eth\thorp s/\thorp s$. Since $\eth A=0$ we can just as well write $\tau = \eth(A\thorp s)/(A\thorp s)$. The quantity $A\thorp s = N^a\nabla_as = \dd s/\dd u$ is therefore the change of $s$ along the null generators with respect to the Bondi time $u$ which is distinguished by the conformal gauge in operation. It turns out that the more natural quantity to use is the reciprocal $\alpha = \dd u/\dd s$. Since $\dd u = \alpha \dd s$ we may call it the \emph{null lapse}, relating the two notions of time along the null generators. Then $\tau$ describes the change in the passage of Bondi time between two different $s$-cuts, $\tau=0$ indicating that the changes in $u$ and $s$ are proportional across a cut.

The null lapse $\alpha$ is a conformal density with weights $[1;0,0]$ and it satisfies the equation $\ethc\alpha = \eth \alpha + \tau \alpha = 0$. Consider now the commutator
\[
  \thorpc\ethc \alpha - \ethc\thorpc \alpha = - \ethc \thorpc \alpha = \cP \alpha.
\]
This equation determines $\cP$ in terms of the chosen system of cuts
\begin{equation}
  \label{eq:34}
  \cP = - \ethc (\thorpc \alpha/\alpha).
\end{equation}
The other quantity $\cQ = K - \etp\tau$ is mostly determined by the geometry of the cuts since $K$ here is real and therefore equals one half of the Gauß curvature of the cuts. The conformal correction $\etp\tau$ is achieved with information about the cut system.

The vanishing of $\cP$ is the integrability condition for the simultaneous equations $\rho' = \thorp \theta$ and $\tau = \eth\theta$. Thus, when $\cP=0$ we can find a conformal factor $\Theta$ such that both $\rho'$ and $\tau$ vanish after a conformal rescaling. If $\cP$ does not vanish, then only one of those quantities can be made to vanish. This is usually taken to be the convergence $\rho'$ whose vanishing indicates that the metrics induced on the cuts all agree. The remaining non-vanishing of $\tau$ then implies that the chosen cuts are not Bondi cuts, i.e., do not belong to a Bondi system.

In the remainder of the paper we will occasionally refer to two special gauges. The first is what we call the \emph{cylinder gauge} which is characterised by the vanishing of $\rho'$ and choosing the unit-sphere metric as the common metric of the $s$-cuts, so that $K=\frac12$ is constant. The other gauge that we use is a Bondi system which makes the additional assumption that $\tau=0$, i.e., that the null lapse is constant across a cut.

\section{The BMS algebra}
\label{sec:bms-algebra}

An important role in the discussion of the structure of $\scri$ is played by its invariance group, the so called BMS group. This is the group of diffeomorphisms of $\scri$ which leave the universal structure tensor~\eqref{eq:6} invariant. Here, we determine its algebra, i.e., the Lie algebra $\sB$ of infinitesimal BMS generators as the vector fields $X^a$ on $\scri$ for which
\[
  \Lie{X}{\Gamma_{ab}{}^{cd}} = 0.
\]
We write a general vector field $X^a$ as a linear combination of the tetrad vectors tangent to $\scri$
\[
  X^a = \eta n^a + \xib m^a + \xi \mb^a,
\]
where $\eta$ and $\xi$ are conformal densities $\eta:[0;1,1]$ and $\xi:[1;1,-1]$.

A short calculation yields the expressions
\begin{equation}
  \label{eq:35}
  \begin{aligned}
    \Lie{X}{g_{ab}} &= - g_{ab} \left( \etp \xi + \eth \xib - 2\rho'\eta \right) - 2 \mb_a \mb_b\eth \xi - 2 m_a m_b\etp \xib \\
    &\qquad\qquad+ 2l_{(a}\mb_{b)} \left(\thorp \xi + \rho' \xi \right) + 2l_{(a}m_{b)} \left(\thorp \xib + \rho' \xib \right) ,\\
    \Lie{X}{N^a} &= \left( \eta \thorp A - A \thorp \eta + A \tau \xib + A \taub \xi \right) n^a - A \left( \thorp \xib + \rho' \xib  \right) m^a - A \left( \thorp \xi + \rho' \xi  \right) \mb^a.
  \end{aligned}
\end{equation}
The invariance of $\Gamma_{ab}{}^{cd}$ implies that
\[
  \Lie{X}{g_{ab}} N^c  + 2 g_{ab} \,\Lie{X}{N^c} = 0.
\]
Thus,  $\Lie{X}{g_{ab}} \propto g_{ab}$ and $\Lie{X}{N^a} \propto N^a$ so that the off diagonal terms in $\Lie{X}{g_{ab}}$ must vanish as do the terms in $\Lie{X}{N^a}$ not proportional to $n^a$. This means that we get the equations
\[
\thorp \xi + \rho' \xi = 0, \qquad \eth\xi = 0  
\]
together with
\[
  A\etp \xi + A\eth \xib - 2\rho'\eta + 2 (\eta \thorp A - A \thorp \eta + A \tau \xib + A \taub \xi ) = 0.
\]
Taken together we find that a BMS generator $X^a$ is represented by a pair of two conformal densities $X=(\xi,\eta_\xi)$ subject to the equations
\begin{equation}
  \label{eq:36}
  \ethc\xi = 0, \qquad \thorpc\xi = 0, \qquad \thorpc\eta_\xi = \frac12\left( \etbc\xi + \ethc\xib\right).
\end{equation}
Using these equations the commutator between two BMS generators $X$ and $X'$ becomes
\begin{equation}
  [(\xi,\eta),(\xi',\eta')] = \left(\xi\diamond \etbc\xi, \frac12\eta\diamond \etbc\xi + \frac12\eta\diamond \ethc\xib + \xi\diamond\etbc\eta + \xib\diamond\ethc\eta \right)\label{eq:37}
\end{equation}
where $\alpha\diamond\beta = \alpha \beta' - \beta \alpha'$.

It is straightforward to verify that generators of the form $(0,\eta)$ (for which $\thorpc\eta=0$) form a sub-algebra which is in fact an Abelian ideal $\sS$, since $[\sB,\sS] \subset \sS$. The elements of $\sS$ are called (infinitesimal) super-translations. They are represented here as conformal densities $\eta$ on $\scri$ satisfying $\thorpc\eta=0$, thus corresponding to global functions on the sphere of null generators of $\scri$. 

On the other hand, given any cut of $\scri$ the Lie algebra of its isotropy group can be given as $(\xi,\eta_\xi)$ where $\eta_\xi$ vanishes on the cut. Each of these algebras is isomorphic to the Lie algebra $\mathfrak{so}(1,3)$ of infinitesimal Lorentz transformations, a real form of $\mathfrak{sl}(2,\CC)$ (see app.~\ref{sec:asympt-symm-1}). \jf{Say something about only considering the complex form} However, there is no distinguished unique sub-algebra of $\sB$, isomorphic to the Lorentz generators.

Sachs~\cite{Sachs:1962b} has shown that the infinite dimensional ideal of super-translations contains a 4-dimensional sub-ideal $\sT$, the (infinitesimal) translations. In the present context we can obtain this ideal by observing that $\sT$ must be invariant under the adjoint action of $\sB$ on $\sS$. A generator $(\xi,\eta_\xi)\in\sB$ acts on $(0,\eta)$ according to the adjoint action $(0,\eta)\mapsto (0,\ad_\xi\eta):=[(\xi,\eta_\xi),(0,\eta)]$ (in a blatant abuse of notation) where
\[
  \ad_\xi\eta =  \xi\etbc\eta - \frac12\eta\etbc\xi.
\]
A closer look reveals (see app.~\ref{sec:asympt-symm-1}) that $\sT\subset\sS$ is spanned by eigenvectors of $\ad_{\xi_0}$ for some generator $\xi_0$, i.e., conformal densities satisfying\jf{Mention reality condition and spin-1/2 representation}
\begin{equation}
  \xi_0\etbc\eta - \frac12\eta\etbc\xi_0 = \lambda_0 \eta,\label{eq:38}
\end{equation}
where $\lambda_0=\pm\frac12$. This characterisation of the space of asymptotic translations is not very convenient since it makes reference to the Lorentz part $\xi$ of generators. It is possible to give another characterisation of $\sT$ which is much more useful.
\begin{thm}
  Let $\B$ be the conformal density for which
  \begin{equation}
  \etbc\B+\ethc\cQ=0,\label{eq:39}
\end{equation}
  and let $U$ be a conformal density with weights $[1;0,0]$ then
  \begin{equation}
  \ethc[2]U = \B U,\label{eq:40}
\end{equation}
if and only if $\eta=A U\in \sT$.
\end{thm}
\begin{proof}[Proof]
  Every $\eta\in\sT$ is a linear combination of eigenfunctions of the operator $\ad_{\xi_0}$ defined by the left hand side in~\eqref{eq:38}. It is enough to focus only on those. Thus, we assume that $\eta$ satisfies \eqref{eq:38} with $\lambda=\pm\frac12$. Applying $\eth_c^2$ to this equation and using the commutator relations between $\ethc$ and $\etbc$ yields
  \begin{equation}
    \label{eq:41}
      \xi_0\etbc\ethc[2]\eta - \frac12\ethc[2]\eta\etbc\xi_0 = \lambda_0 \ethc[2]\eta + \xi_0\ethc\cQ\eta.
  \end{equation}
  Since $\etbc A = 0$, the conformal density $U=A\eta$ satisfies the same equation. Defining $Z:=\ethc[2]U-\B U$ and inserting into~\eqref{eq:41} yields
  \[
      \xi_0\etbc Z - \frac12 Z\etbc\xi_0 = \lambda_0 Z + \xi_0\ethc\cQ U + \xi_0\etbc\B U.
  \]
Since $\B$ satisfies~\eqref{eq:39} this shows that $Z$ satisfies~\eqref{eq:38} as well. Referring to app.~\ref{sec:equat-co-curv} and noting that $Z$ has spin-weight $s=2$ implies that $Z=0$ identically. Thus, $\eta\mapsto A\eta$ defines a linear map from $\sT$ to the solution space of~\eqref{eq:40} which is injective. Since both spaces are 4-dimensional this establishes the result.
\end{proof}
Note, that in view of this result we will henceforth identify $\sT$ with the solution space of~\eqref{eq:40}, i.e., an asymptotic translation is a conformal density $U:[1;0,0]$ satisfying~\eqref{eq:40}.

The conformal density $\B$ introduced above is closely related to the ``gauge field'' $\rho_{ab}$ defined by Geroch~\cite{Geroch:1977}. As shown in app.~\ref{sec:equat-co-curv} it is uniquely determined by the equation $\etbc\B+\ethc\cQ=0$. In view of this equation and in lack of a better name we call the field $\B$ the \emph{co-curvature}. This equation also fixes the behaviour of the co-curvature along the null generators as follows from the Jacobi identity~\eqref{eq:32} and
\[
  \begin{multlined}
  0 = \thorpc(\etbc\B + \ethc\cQ) = \etbc\thorpc\B + 2 \cP \cQ + \ethc \thorpc\cQ\\ = \etbc\thorpc\B + 2 \cP \cQ - \ethc \etbc\cP = \etbc(\thorpc\B - \ethc\cP).
\end{multlined}
\]
Since $\etbc$ acting on densities with weight $[-2;1,-3]$ is an isomorphism (see app.~\ref{sec:equat-co-curv}) we find
\begin{equation}
  \label{eq:42}
  \thorpc\B = \ethc\cP.
\end{equation}

\section{The Lorentz metric on the space of translations}
\label{sec:lorentz-metric-space}

Our next problem is to define a metric on $\sT$. This is necessary for several reasons, the most important one being that the energy-momentum is commonly defined as a co-vector on $\sT$ and it is interpreted as a Lorentzian 4-vector. This implies that there is Lorentzian metric on $\sT^*$, the dual space to $\sT$ and hence also on $\sT$. Furthermore, it is important to normalise the translations in order to fix the magnitudes of the computed energies and momenta.

A discussion of this issue is difficult to find in the literature since in the Bondi gauge the first four spherical harmonics are implicitly taken as an orthonormal Lorentzian basis for $\sT$ with respect to which the energy-momentum vector can be computed. This is consistent with spinorial approaches to the Bondi energy-momentum based on the Nester-Witten form (see~\cite{Horowitz:1982,Szabados:2009}) where the structure of spin-space is used to define a Lorentzian metric on $\sT$. However, at least to our knowledge there is no explicit definition of a Lorentzian metric on $\sT$ defined as the solution space of~\eqref{eq:40}.

A hint about how to proceed in the general case can be found in a paper by Hansen, Janis et al~\cite{Hansen:1976} where a representation of a Minkowski vector at a point $P$ as a function on the light cone of $P$ is discussed. Given a vector $U^a$ at $P$, the authors associate a scalar $U = U^al_a$ to every null vector $l^a$ at $P$ with $t^al_a=1$, where $t^a$ is a future-pointing time-like unit-vector, the first member of an orthonormal basis at~$P$. This defines a function $U$ on the unit-sphere of null vectors $l^a$ at $P$ which automatically satisfies the equation $\eth^2U=0$. For any two vectors $U^a$ and $V^a$ at $P$ and their representative functions $U$ and $V$ on the sphere one can define\footnote{Note, that the conventions regarding $\eth$ in~\cite{Hansen:1976} are different from the ones used here.}
\begin{equation}
  G[U,V] = UV + U \eth\etp V  + V \eth\etp U - \eth U \etp V - \eth V \etp U\label{eq:43}
\end{equation}
and, remarkably, one finds that $G[U,V]=U_aV^a$, i.e., the expression $G[U,V]$ is constant on the sphere and evaluates to the inner product between the two vectors. 

Taking this expression as a starting point we consider now a more general sphere $\cS$ in Minkowski space and two functions $U$ and $V$ on it. We first assume the sphere to have a constant radius $R$. Then we find that for dimensional reasons we should consider the expression
\[
  G[U,V] = \frac1{R^2} UV + U \eth\etp V  + V \eth\etp U - \eth U \etp V - \eth V \etp U.
\]
Applying $\eth$ to this equation and using the commutator equation $[\eth,\etp]\alpha = -sR^{-2} \alpha$ for any quantity $\alpha$ with spin-weight $s$ on a sphere of radius $R$ we obtain
\[
  \eth G[U,V] =  - \eth^2 U\, \etp V - \eth^2 V\, \etp U,
\]
showing that $G[U,V]$ is constant on $\cS$ if and only if $U$ and $V$ satisfy the equation $\eth^2U=0=\eth^2V$.

Next, we consider the general case. Again, we need to make adjustments to the formula. This time, it is natural to replace the $R^{-2}$-term with the Gauß curvature of the sphere and in view of the conformal invariance we use the conformal density $\cQ$. Furthermore, we consider the quadratic form $q[U]:=G[U,U]$ since the bilinear form can be obtained by polarisation and we use the general conformally invariant derivatives. This puts our focus onto the quadratic expression (and its associated bilinear form $G[U,V]$)
\begin{equation}
\label{eq:44}
  q[U] = 2\cQ U^2 + 2U \ethc\etbc U - 2\ethc U \etbc U
\end{equation}
defined  in terms of a scalar $U$ on $\scri$. First, we note that this expression is a conformal density and if $U$ has weights $U:[1;0,0]$ then $q[U]$ is a scalar, i.e., all its weights vanish. Furthermore, a short calculation shows that
\[
  \ethc q[U] = 2\etbc\left( U \left[\ethc[2]U - \B U\right]\right).
\]
In a similar way we compute
\[
  \thorpc q[U] = 2G[\thorpc U,U]
\]
This shows that $q[U]$ is constant on a $s$-cut if and only if $\ethc[2]U = \B U$, since $\etbc \eta=0$ has only the trivial solution for $\eta$ with weights $[0;2,-2]$\jf{Mention this in the appendix.}. Since $q[U]$ is conformally invariant we can evaluate it in any gauge. Doing this in a Bondi system reduces $q[U]$ to $G[U,U]$ with $G$ as given in~\eqref{eq:43} on any Bondi cut. Since this is a non-degenerate bilinear form, we find that $\thorpc q[U] = 0$ if and only if $\thorpc U=0$. Thus, we have shown that $q[U]$ is constant on $\scri$ if and only if $U\in\sT$, so that $q[U]$ is a quadratic form with Lorentzian signature on $\sT$ which we take as the definition of the Lorentz metric on the space $\sT$ of asymptotic translations.

When rewriting the quadratic form with respect to the usual $\eth$ and $\etp$-operators we obtain
\begin{equation}
  q[U] = 2K U^2 + 2 U \eth\etp U - 2 \eth U \etp U.\label{eq:45}
\end{equation}
Now recall that $K=\frac12 k[g]$, where $k[g]$ is the Gauß curvature of the induced metric $g_{ab}$ on the cut through the point where $q[U]$ is evaluated and that the Gauß curvature transforms under a conformal rescaling $g_{ab} \mapsto \Theta^2 g_{ab}$\jf{Should one use a different symbol for the 2-metric?} as follows
\[
  k[\Theta^2g] = \Theta^{-2}\left(k[g] - 2\Theta\, \eth\etp \Theta + 2\eth\Theta\, \etp\Theta \right).
\]
Rewriting this in terms of $K$ we get the equation
\begin{equation}
  \label{eq:46}
    \Theta^2K[\Theta^2g] + \Theta\, \eth\etp \Theta - \eth\Theta\, \etp\Theta = K[g] .
\end{equation}
Suppose now that $U$ does not vanish anywhere, then we can replace $g$ with $U^{-2}g$ and, comparing with \eqref{eq:45}, we find
\begin{equation}
  q[U] = 2 (U^2K[g] + U\, \eth\etp U - \eth U\, \etp U) = 2 K[U^{-2}g].\label{eq:47}
\end{equation}
Therefore, we find that $q[U]=1$, i.e., $U$ is normalised, if and only if $U^{-2}g_{ab}$ has Gauß curvature equal to $1$, i.e., it is the metric on the unit-sphere.

\section{The mass aspect}
\label{sec:mass-aspect}

We are now ready to discuss the main ingredient in the definition of the Bondi energy-momentum, the mass aspect. We take as our starting point the definition given in \cite{Penrose:1986}
\[
  \fM = A^{-1}\left(\sigma \Phi_{20} - A \psi_2\right).
\]
As it stands, this expression is neither real nor conformally invariant. We address both issues next.

As mentioned in sec.~\ref{sec:asympt-flatn-struct} the imaginary part of \eqref{eq:19} is intrinsic to $\scri$ and it yields the following equation
\[
  -\etp\Phi_{01} + \eth\Phi_{10} = \sigmab \Phi_{02} - \sigma \Phi_{20} +A \psi_2 - A \psib_2,
\]
which after inserting \eqref{eq:9} becomes
\begin{equation}
  \label{eq:48}
  - A \psi_2 + \sigma \Phi_{20} + \etp^2\sigma =    - A \psib_2 + \sigmab \Phi_{02} + \eth^2\sigmab  .
\end{equation}
so that the addition of the $\etp^2\sigma$ term yields a real quantity.

However, this expression is still not conformally invariant. As it stands it is the result of evaluating a conformal density in a Bondi system which is defined by $\rho'=0$ and $\tau=0$. We can find out what the full expression must look like by transforming to a different gauge and observing which terms appear. This calculation suggests to define a conformal density $N$ of weights $[-2;-2,2]$ by
\begin{equation}
  \label{eq:49}
  N = \Phi_{20} - \rho' \sigmab - \etp\taub + \taub^2.
\end{equation}
We will see below that it is closely related to Bondi's \emph{news function}. Now we can define the mass-aspect $\fM$ as a conformal density $\fM:[-3;0,0]$ by the equation 
\begin{equation}
\label{eq:50}
  A \fM =  - A\psi_2 + \sigma N + \etbc[2]\sigma.
\end{equation}
Note the appearance of the conformal $\etbc$ operator here. The advantage of writing $\fM$ in this way is that it becomes manifestly conformally invariant since each term is a conformal density. However, when rewritten in terms of the usual $\etp$-operator then the mass-aspect becomes almost the familiar expression, namely
\[
  \fM =  A^{-1}\left(- A\psi_2 + \sigma \Phi_{20} + \etp^2\sigma - \rho' \sigma\sigmab \right).
\]
The next step is to compute the derivative of $\fM$ along the generators. To do this we need to find the $\thorpc$-derivatives of the quantities defining the mass-aspect. The easy ones are
\[
  \thorpc A = 0,\qquad \thorpc \sigma  = -\Nb.
\]
Next we compute the derivatives of $N$. After a calculation using \eqref{eq:23} we find the result
\begin{equation}
  \label{eq:51}
  \thorpc N = A \psi_4 - \etbc \Pb.
\end{equation}
Similarly, using the difference of \eqref{eq:20} and \eqref{eq:21} we find that
\begin{equation}
  \label{eq:52}
  \ethc N = A \psi_3 + \etbc Q.
\end{equation}
The derivative of $\fM$ along the generators can be computed in a straightforward way using the intrinsic equations that we have just established. We find the intriguing formula
\begin{equation}
  \label{eq:53}
   A \thorpc \fM =  - N\Nb  - \ethc[2]N   - \etbc[2]\Nb + \ethc\etbc Q .
\end{equation}

\section{The mass-loss formula}
\label{sec:mass-loss-formula}

At this point we have everything in place to discuss the Bondi-Sachs mass-loss formula.
% We first define the Bondi energy-momentum on some cut $\cC$ of $\scri$. Let $U\in\sT$ be an asymptotic translation then the energy-momentum component defined by $U$ is
We first define the following integral for a given cut $\cC$ of $\scri$ and an arbitrary conformal density $U$ with weights $[1;0,0]$
\begin{equation}
  \label{eq:54}
4\pi   m_\cC[U] = \int_{\cC}U\left(\sigma N + \etbc[2]\sigma - A \psi_2 \right)  A^{-1}\, \dd^2\cS.
\end{equation}
Here, $\dd^2\cS=\i\mbbar\mbb$ is the area 2-form on $\cC$ induced from the unphysical metric $g_{ab}$. Note, that with $U$ of the given type  the integral is well-defined in the sense that the integrand is a conformal density with vanishing weights so that the value is unambiguously defined.

The main achievement of the analysis by Bondi and coworkers\cite{Bondi:1960,Bondi:1962,Sachs:1962a} in the early 1960 is the formulation of the mass-loss formula which has been given in many different formulations since then. We will provide here another point of view on the mass-loss which highlights the tight interplay between the asymptotic translations and the properties of the mass-aspect. Our presentation follows closely the one given in~\cite{Penrose:1986} based on Stokes' theorem as formulated in app.~\ref{sec:stokes-theorem-cghp}.

Consider the 2-form $\mub = \mu_0\, \i\mbbar\mbb + \mu_1 \i \,\lbb\mbb - \mu_2\,\i \lbb\mbbar$ on $\scri$ defined by
\[
  \begin{aligned}
  \mu_0 &= A^{-1} U\left(\sigma N + \etbc[2]\sigma - A \psi_2 \right),\\
  \mu_1 = \bar{\mu}_2 &= A^{-1}\left(U\ethc (N + \tfrac12 \Bb) - \ethc U (N + \tfrac12 \Bb)\right) 
\end{aligned}
\]
then $\mub$ pulled back to any cut $\cC$ is the Bondi energy-momentum relative to $U$ evaluated on that cut. Let $\cC_1$ and $\cC_2$ be any two cuts from the family defined by constant $s$ and integrating $\dd\mub$ over the 3-dimensional piece $\scri_1^2$ between them yields via Stokes' theorem
\[
  m_2[U] - m_1[U] = \frac1{4\pi}\int_{\scri_1^2} \dd\mub .
\]
Using~\eqref{eq:72} we find $\dd\mub$ by computing
\[
  \begin{multlined}
    A\left(\thorpc\mu_0 + \ethc\mu_1 + \etbc\mu_2\right) = \thorpc U A \fM - U N\Nb - U \ethc[2]N - U\etbc[2]\Nb + U\ethc\etbc\cQ \\
    + U \ethc[2](N+\tfrac12\Bb) - \ethc[2]U\, (N + \tfrac12 \Bb)
    + U \etbc[2](\Nb+\tfrac12\B) - \etbc[2]U\, (\Nb + \tfrac12 \B) \\
    =   - U N\Nb  + U(\ethc\etbc\cQ + \tfrac12 \ethc[2]\Bb + \tfrac12 \etbc[2]\B) - U\B \, (N + \tfrac12 \Bb) - U\Bb \, (\Nb + \tfrac12 \B)\\
    - (\ethc[2]U - \B U)\, (N + \tfrac12 \Bb)
    - (\etbc[2]U - \Bb U)\, (\Nb + \tfrac12 \B) + \thorpc U A \fM\\
    =   - U (N\Nb  + \B \, N  + \Bb \,\Nb + \B\Bb)\\
    - (\ethc[2]U - \B U)\, (N + \tfrac12 \Bb)
    - (\etbc[2]U - \Bb U)\, (\Nb + \tfrac12 \B) + \thorpc U A \fM
\end{multlined}
\]
The form of the first term in the final expression suggests to consider the quantity $\news = N+\Bb$. Inspection of \eqref{eq:42}, \eqref{eq:51} and \eqref{eq:52} shows that $\news$ satisfies
\begin{equation}
  \thorpc\news = A \psi_4, \qquad \ethc\news = A \psi_3,\label{eq:55}
\end{equation}
so that $\news$ can be considered as a potential for the part of the Weyl tensor that is entirely intrinsic to $\scri$. Note, that these two simultaneous equations for $\news$ can be solved because their integrability condition is the Bianchi equation~\eqref{eq:33} for $k=3$. We will take this as the defining property for the \emph{Bondi news} $\news$.

Now the calculation above proves the following
\begin{thm}[Bondi-Sachs mass-loss formula]
  Let $U\in \sT$ be an arbitrary translation and let $\cC_1$ and $\cC_2$ be two arbitrary $s$-cuts, then
  \[
    m_2[U] - m_1[U] = -\frac1{4\pi}\int_{\scri_1^2} U \,\news\bar{\news}\, \dd^3V.
  \]
  Furthermore, if $U$ is a time translation (so that $U>0$) then\jf{Maybe talk about that in the appendix}
  \[
    m_2[U] < m_1[U] .
  \]
\end{thm}
Note, that the mass-loss formula holds only if $U$ is an asymptotic translation, and therefore, only in this case are we justified to interpret~\eqref{eq:54} as a component of an energy-momentum. Furthermore, note that $m_\cC[U]$  is linear in $U$ and, therefore, it scales with the 'size' of $U$. Not only for this reason it is important to have a Lorentz metric defined on $\sT$ which allows us to normalise $U$. It also allows us to distinguish temporal from spatial translations. Thus, when $U$ is time-like \eqref{eq:54} defines an energy, while for space-like $U$ we get a momentum component.

Finally, we come back to the mass-aspect. We defined it in \eqref{eq:50} without any reference to asymptotic translations based entirely upon the requirements of reality and conformal invariance. However, in view of the fact that it is integrated against an asymptotic translation we can now also redefine the mass-aspect due to the following calculation
\[
  \begin{multlined}
  \int_\cC U \left(\sigma N + \etbc[2]\sigma - A\psi_2\right)\,A^{-1}\dd^2\cS = 
  \int_\cC U \left(\sigma N  - A\psi_2\right)\,A^{-1} + \sigma\etbc[2]UA^{-1}\dd^2\cS \\
  = \int_\cC U \left(\sigma N  - A\psi_2\right)\,A^{-1} + \sigma\Bb UA^{-1}\dd^2\cS = 
  \int_\cC U \left(\sigma \news  - A\psi_2\right)\,A^{-1}\dd^2\cS .
\end{multlined}
\]
Thus, we could also adopt the definition
\begin{equation}
  \label{eq:56}
  \fM = A^{-1}\left(\sigma \news  - A\psi_2\right)
\end{equation}
for the mass-aspect. This definition looks very similar to the one given at the beginning of sec.~\ref{sec:mass-aspect} except that $N=\Phi_{20}$ is replaced by the news $\news$. Note, that this definition is conformally invariant and, when integrated against an asymptotic translation yields real values. However, the definitions \eqref{eq:50} and \eqref{eq:56} are not identical as conformal densities on $\scri$. They yield the same values only when integrated against an asymptotic translation, i.e., they define the same energy-momentum covector regarded as a linear form on $\sT$.

\section{Computing the Bondi energy-momentum}
\label{sec:comp-bondi-energy}

In this final section we want to briefly discuss how one would go about explicitly calculating the Bondi energy and momenta when faced with a representation of $\scri$ which is not a Bondi system. Let us assume that we are given 3-dimensional null hypersurface of a space-time given in the form of a family of 2-dimensional cuts labeled by a scalar $s$. We are also given a null tetrad which we can assume to be adapted to the cuts in the sense that $n^a$ points along the null generators and $m^a$ is tangent to the cuts. This can always be achieved by at most two null rotations. On each cut we have all the spin-coefficients and curvature quantities available, all given with respect to some generic fixed coordinate system.

This situation occurs in numerical relativity, when one solves equations which allow full access to null infinity but which deny the possibility to introduce gauges adapted to $\scri$ because of some numerically relevant considerations. An example which was the motivation for this work is described in~\cite{Beyer:2017}.

In order to compute an energy-momentum component on a given cut $\cC$ one needs two main ingredients, the mass-aspect and an asymptotic translation $U$. The latter is more involved to compute so we describe this first. We need to select an appropriate solution of the equation~\eqref{eq:39}. The way we proceed is to first compute a conformal factor $\Omega$ which scales the unit-sphere metric to the metric on $\cC$. It must satisfy the equation
\begin{equation}
  \label{eq:57}
  \Omega^2 K + \Omega\, \eth\etp \Omega - \eth \Omega\, \etp \Omega = 1 .
\end{equation}
where the $\eth$-operator is defined in terms of the given null tetrad and the corresponding spin-coefficients and with $K$ also being computed from the given data on the cut. This is a non-linear elliptic equation for $\Omega$ which has a unique solution.

The conformal factor $\Omega$ defines a conformal density with weights $[1;0,0]$ so that we can write every asymptotic translation in the form $U=\Omega V$ with a scalar function $V$. Furthermore, from~\eqref{eq:47} we have $q[\Omega] = 1$ so that $\Omega$ satisfies~\eqref{eq:39}. Inserting $U$ into~\eqref{eq:39} we obtain an equation for $V$
\[
2\ethc\Omega\ethc V + \Omega\ethc[2]V =  0.
\]
Rewriting this equation in terms of the usual $\eth$-operator on $\cC$ yields
\begin{equation}
  \label{eq:58}
  \eth\left(\Omega^2\eth V\right) = 0.
\end{equation}
This equation has four linearly independent solutions, one being $V=1$, corresponding to $\Omega$ being a solution of~\eqref{eq:39}. It defines a time-like translation. The other three solutions can be chosen to correspond to translations in three mutually orthogonal space-like directions. To this end one needs the scalar product on $\cT$ in order to produce an orthonormal basis for the solution space.

There is one practical aspect related to the use of the quadratic form $q[U]$. As  defined in \eqref{eq:44} or \eqref{eq:45} it is a constant \emph{function} on $\scri$. For practical purposes it is much better to compute a single number. This is most easily done by integrating the function over the full cut. Let $|\cC|$ denote the area of the cut then, since $q[U]$ is constant for all asymptotic translations, we have
\[
  q[U] = \frac2{|\cC|} \int_{\cC} K U^2 - 2 \eth U\, \etp U\; \dd^2\cS
\]
after an integration by parts. Note that, in a Bondi frame where $K=\frac12$ and for the time translation given by $U=1$ this formula yields $q[U]=1$ due to the Gauß-Bonnet theorem.

Once the translations have been selected one needs to compute the momenta by integrating them against the mass-aspect. In principle, this can be done by using any of the two definitions of the mass-aspect but both of those have disadvantages from a practical point of view: the definition in \eqref{eq:56} involves $\B$ which must be obtained by solving an equation while the definition in \eqref{eq:50} involves a second order derivative. To avoid both issues one can also use the integral
\[
  m_\cC[U] = \int_{\cC} U N (A^{-1}\sigma)  - U \psi_2  - \etbc U \etbc (A^{-1}\sigma) \, \dd^2\cS,
\]
which is obtained from \eqref{eq:50} by integrating by parts.

\section{Discussion}
\label{sec:conclusion}

In this work we have presented a treatment of null infinity and the Bondi energy-momentum which is manifestly conformally invariant and makes no reference to any special gauge. Most discussions of the subject jump very quickly to the introduction of a simplifying gauge (which always exists). The only exception seems to be Geroch's discussion~\cite{Geroch:1977} of null infinity (see also \cite{Tafel:2000} for related questions). For us, the motivation to revisit the issue of the Bondi energy-momentum and the interplay with the asymptotic translations arose from the need to compute the Bondi mass under circumstances which are far away from any simplifying gauges. In a related paper~\cite{Frauendiener:2021b} we present our numerical implementation of the present theoretical discussion in an application where we study the response of a black hole to the impact of gravitational waves.

We hope that our treatment of null infinity offers some new insights into the subject. In particular, the characterisation of the Bondi news ``function'' as a potential for the intrinsic part of the Weyl tensor is very simple and gauge independent. In our treatment of the topic we were guided by the notion of conformal invariance. It might be useful to spell out here what is meant by this term. To be sure, the notion of the Bondi-Sachs energy-momentum is a property of the physical space-time and, therefore, only defined for the physical metric. It is not conformally invariant in the sense of giving the same values when the metric is rescaled by an arbitrary conformal factor. Instead, we represent the physical metric $\tilde g_{ab}$ in terms of a conformal factor $\Omega$ and a metric $g_{ab}$ (see~\eqref{eq:1}) as $\tilde g_{ab}=\Omega^{-2}g_{ab}$. Clearly, all physical properties attributed to the metric $\tilde g_{ab}$ should be independent of the specific choice of $\Omega$ and $g_{ab}$ made to represent it, i.e., they should be invariant if $(g_{ab},\Omega)$ is replaced by $(\theta^2g_{ab},\theta\Omega)$ for arbitrary non-vanishing functions $\theta$. 

The Lorentzian metric on the space of asymptotic translations seems to have been unknown until now, at least in its explicit form as given here. It is present in all treatments of $\scri$ albeit mostly implicitly, being defined in terms of a basis for the translations given by the first four spherical harmonics.

The appearance of the co-curvature $\B$ is crucial for the consistency of the entire structure on $\scri$. On the one hand, its presence should not surprise us since Geroch had already introduced the corresponding tensorial quantity. However, on the other hand, its importance may have been underestimated until now. It appears in our treatment essentially as an auxiliary field that mediates between $\ethc$ and $\etbc$ equations on $\scri$ but clearly it is intimately tied in with the geometric properties of a 2-surface and its embedding into a null hypersurface. It would not be surprising if it turned out that the co-curvature also plays a major role in a satisfactory definition of quasi-local energy-momentum \cite{Szabados:2009,Penrose:1982,Penrose:1984}. This remains to be seen.

Finally, let us comment on the cGHP formalism. In our discussion we have made the assumption that the null tetrad on $\scri$ is adapted to some family of cuts of $\scri$. The reason for this assumption was merely to simplify the equations because it implies that $\rho$ is real; it is by no means necessary. In fact, there might be an advantage when the only condition imposed on the frame is that $n^a$ be aligned with the null generators of $\scri$ and leaving the freedom of a null rotation around that vector. Then the null congruence generated by the transverse null vector $l^a$ is no longer surface forming and $\rho$ acquires an imaginary part. This is the setup that is necessary to discuss the remarkable results by Newman and coworkers~\cite{Adamo:2012} about equations of motion of (charged) particles in a relativistic field being obtained from the asymptotic structure of null infinity alone. It is hoped that the comparative simplicity of the cGHP formalism (at least on $\scri$) may help to shed some new light onto these developments which are hidden behind some rather complicated calculations.

\appendix

\section{Explicit calculations}
\label{sec:expl-calc}

We need to study equations on $\scri$ involving the cGHP operators $\ethc$, $\etbc$ and $\thorpc$ acting on conformal densities. The equations are manifestly conformally invariant, so that they can be evaluated in any gauge. The most useful gauges are, as mentioned in sec.~\ref{sec:cut-systems-null} the cylinder gauge and the Bondi gauge. In both gauges the convergence $\rho'$ is chosen to vanish, resulting in the same metric being defined on any cut, which is taken to be the unit-sphere metric. The Bondi gauge imposes the additional condition $\tau=0$ which has the consequence that the cuts ``move'' uniformly into the future with increasing $s$.

For the explicit calculations below we choose the cylinder gauge and, in addition, on the unit-sphere we choose stereographic coordinates $(z,\zb)$ so that the metric reads\jf{2-metric?}
\[
  g = - \frac{2 \dd z\, \dd\zb}{P^2} ,\qquad \text{ where } P = \frac1{\sqrt2}(1+z\zb).
\]
This yields the standard expressions for the $\eth$ and $\etp$ operators as found in~\cite{Penrose:1984a} acting on a quantity $\eta$  with weights $[w;s,-s]$
\begin{equation}
  \label{eq:59}
  \eth \eta = P^{1-s}\del_\zb(P^s \eta), \qquad 
  \etp \eta = P^{1+s}\del_z(P^{-s} \eta).
\end{equation}

The general strategy is to reformulate the conformally invariant equations so that they can be written in terms of these operators. This can be most easily achieved by using the frame factor $A$ and the null lapse $\alpha$. Both are annihilated by $\ethc$ and $\etbc$ on a cut. This means that they can be used to alter the weights of the quantities on which these operators act. For example, for $\eta:[w;p,q]$ the density $\tilde\eta = A^{-\frac12(p+q)} \alpha^{-(w+q)}$ has weights $[s;s,-s]$ with $s=\frac12(p-q)$ so that
\[
  \ethc \eta = 0 \iff \eth \tilde\eta = 0.
\]
In this way we try to rephrase the conformally invariant equations into equations involving the standard $\eth$ operator and then study those. In the following sections we consider most of the equations that were mentioned in the main text.

A major role is played by the requirement that the quantities in question must be globally defined on the sphere. Without going into details (see~\cite{Eastwood:1982}) we just remind  the reader that a spin-weight $s$ quantity represented by a globally regular function $\eta(z,\zb)$ in terms of stereographic coordinates $(z,\zb)$ (with respect to the north pole, say) is globally defined on the sphere if and only if its representation in stereographic coordinates with respect to the south pole is globally regular as well, i.e., if and only if
\[
\tilde\eta (v,\vb) =  \eta(v^{-1},\vb^{-1})v^s\vb^{-s} 
\]
is regular for all $(v,\vb)$.

\subsection{Asymptotic symmetries}
\label{sec:asympt-symm-1}

We start with the equations~\eqref{eq:36} determining the BMS generators. The structure of these equations and the commutator~\eqref{eq:37} suggests to consider first complex valued densities $\eta$ subject to the equation $\thorpc\eta = \frac12\etbc\xi$ and complex linear combinations of such pairs $(\xi,\eta)$. This will yield the complexification of the BMS algebra. The real version can then be obtained by regarding $(\xi,\eta_\xi)$ and $(\i\xi,\i\eta_\xi)$ as linearly independent over the reals. The determination of the structure of the BMS algebra can be achieved in a straightforward way in a Bondi system and that is enough to understand the algebraic properties. Our goal here, however, is to relate the BMS generators, in particular the characterization of the asymptotic translations, directly to the properties of the mass-aspect which are expressed in terms of the cGHP formalism. 

With $\xi$ having weights $[1;1,-1]$ and since $\rho'=0$ in the cylinder gauge, the equation $\thorpc \xi=0$ implies that $D'\xi=0$, i.e., that $\xi$ is constant along the null generators of $\scri$. Furthermore, the equation $\ethc\xi=0$ translates directly into the equation
\[
  \del_\zb(P \xi) = 0  \implies \xi(z,\zb) = \frac{h(z)}{1+z\zb} 
\]
where $h(z)$ is an entire function on $\CC$. Since $\xi$ is globally defined on the sphere and has spin-weight $s=1$ the function
\[
  \xi(v^{-1},\bar{v}^{-1}) v \vb^{-1} = \frac{h(v^{-1})}{1+v\vb}v^2.
\]
must be regular for all $v\in\CC$ as well. Therefore, $h$ can only be a complex polynomial of degree at most 2, and this implies that $\xi$ can be regarded as a linear combination of three functions $\xi_0,\xi_\pm$ defined by
\begin{equation}
  \label{eq:60}
  \xi_- = \frac{1}{P}, \quad  \xi_0 = \frac{z}{P}, \quad \xi_+ = \frac{z^2}{P}.
\end{equation}

As discussed previously, the component $\eta$ along the null generators of $\scri$ determined by these three different $\xi$'s is obtained from the equation $\thorpc\eta = \frac12\etbc\xi$, which is translated into the equation
\[
  D'(A^{-1}\eta) = \frac1{2A}\etp\xi + A^{-1}\xi \taub.
\]
Using $D' = n^a\nabla_a = A^{-1}N^a\nabla_a = A^{-1}\del_u = A^{-1}\alpha^{-1}\del_s$ gives
\[
  \del_s(A^{-1}\eta) = \frac{\alpha}{2}\etp\xi - \xi \etp\alpha.
\]
This equation can be integrated by observing that $\int_0^s\alpha(t,z,\zb)\,\dd t$ defines the function $u(s,z,\zb)$, the amount of Bondi time $u$ passed along the null generator given by $(z,\zb)$. Since $\xi$ and $\etp\xi$ are constant along the null generators this yields the general solution with $\eta_0(z,\zb)$ an arbitrary function on the sphere
\[
  \eta(s,z,\zb) = A\eta_0(z,\zb)+\frac{A}{2}\etp\xi\,u(s,z,\zb) - A\xi \etp u(s,z,\zb).
\]
Thus, the general form of a (complex) BMS generator written as a vector field on $\scri$ is (remembering that $An^a = N^a$)
\begin{equation}
  X^a = \xi \mb^a + \left(\frac{u}{2}\etp\xi - \xi \etp u\right)\, N^a + \eta_0 \,N^a,\label{eq:61}
\end{equation}
where $\xi$ is a complex linear combination of the basis given in~\eqref{eq:60}.

For completeness we list a basis for the BMS algebra in a Bondi system, where $s=u$, and in stereographic coordinates
\begin{equation}
  \label{eq:62}
  \begin{aligned}
  X_- &= \del_z  - \frac{\zb}{1+z\zb} u\del_u ,\\
  X_0 &= z\del_z + \frac{1-z\zb}{2(1+z\zb)} u\del_u,\\
  X_+ &= z^2\del_z + \frac{z}{1+z\zb} u\del_u,
\end{aligned}
\qquad
X_\eta = \eta(z,\zb) \del_u,
\end{equation}
where now $\eta$ is an arbitrary function on the sphere. The commutators are
\begin{equation}
\label{eq:63}
\begin{aligned}
[X_0,X_-] &= - X_-,\\
[X_+,X_-] &= \;2 X_0,\\
[X_0,X_+] &= \phantom{-} X_+,
\end{aligned}
\qquad
\begin{aligned}
[X_-,X_\eta] &= 0,\\
[X_0,X_\eta] &= 0,\\
[X_+,X_\eta] &= 0,
\end{aligned}
\qquad
[X_\eta,X_{\eta'}] = 0.
\end{equation}
The commutators among $(X_-,X_0,X_+)$ are those for $\mathfrak{sl}(2,\CC)$, the complexification of $\mathfrak{so}(1,3)$, while the super-translations $X_\eta$ commute with everything.

\subsection{The translations}
\label{sec:translations}

Fix $\xi:[1;1,-1]$ with $\ethc\xi = 0$ and consider the equation
\begin{equation}
  \xi \etbc \eta - \frac12 \eta \etbc \xi = \lambda \eta\label{eq:64}
\end{equation}
for an arbitrary complex-valued conformal density $\eta:[w;p,q]$ and $\lambda\in\CC$. This equation can be rewritten in the form
\[
  \etbc(\xi \eta^{-2}) = -2\lambda \eta^{-2}.
\]
Multiplying this equation with appropriate powers of $A$ and $\alpha$ we find
\[
  \etp(\xi f)  = -2\lambda f, \qquad \text{ with } f = \eta^{-2}\alpha^{-2(1-w-p)}A^{p+q}.
\]
This is an eigenvalue equation for the conformal density $f$ which has weights $[-t;t,-t]$ where $t=1-p+q$. In the cylinder gauge this equation becomes
\[
  P^{1+t}\del_z\left(P^{-t} \xi f\right) = -2\lambda f
\]
and with the explicit expressions for $\xi$ from~\eqref{eq:60} we get the final equation for the function $g=P^{-(1+t)}f$
\begin{equation}
  \del_z (z^m g) = -2\lambda g, \qquad m=0,1,2.\label{eq:65}
\end{equation}
We solve this equation for each $m$ separately, starting with $m=1$. In that case the general solution is
\[
  g(z,\zb) = h(\zb) z^{-2\lambda-1},
\]
with some arbitrary function $h(\zb)$. We search for solutions $\eta$ which are globally defined functions with spin-weight $s=\frac12(p-q)$ on the sphere. Since both $\alpha$ and $A$ are globally defined, this implies that the function $\tilde\eta = \eta\alpha^{(1-w-p)}A^{-(p+q)/2}$ for which
\[
  \tilde\eta(z,\zb) = h(\zb) z^{\frac12+\lambda}P^{s-1},
\]
(where we have redefined the arbitrary function $h(\zb)$) must be globally defined with spin-weight $s$ as well. This implies that $\frac12+\lambda$ is a non-negative integer and that $h$ is an anti-holomorphic function, i.e., it is analytic in $\zb$ for all values of $\zb$. Furthermore, since $\tilde\eta$ has spin-weight $s$ the quantity
\[
  \tilde\eta(v^{-1},\vb{}^{-1}) v^s \vb{}^{-s}
\]
must be globally defined for all $(v,\vb)$, as well. This implies that
\[
  h(\vb{}^{-1})(1+v\vb)^{s-1} v^{\frac12-\lambda} \vb{}^{1-2s}
\]
is globally defined. This is only possible if $\vb{}^{1-2s}h(\vb{}^{-1})$ is globally regular, i.e., if $h$ is a polynomial with degree $n$ where $0\le n \le 1-2s$, and if, furthermore, $\lambda=\pm\frac12$. Thus, for $s>\frac12$ there is no non-trivial solution while for $s\le\frac12$ we find  the eigenvalues $\lambda=\pm\frac12$ with the corresponding eigenfunctions
\[
  \lambda=-\frac12: \quad  \tilde\eta = P^{s-1} \sum_{k=0}^{1-2s}a_k \zb^k;
  \qquad
  \lambda= \frac12: \quad  \tilde\eta = z P^{s-1} \sum_{k=0}^{1-2s}a_k \zb^k.
\]
The case $m=0$ is treated in a very similar way. In that case we find that
\[
   \tilde\eta(z,\zb) = h(\zb) \e^{\lambda z}P^{s-1}
\]
must define a global spin-weight $s$ quantity on the sphere. This leads to $\lambda=0$ and to the same regularity condition for $h$ as above, so that solutions of~\eqref{eq:65} with $m=0$ exist only for $\lambda=0$ and $s<\frac12$, implying
\[
\tilde\eta =   P^{s-1} \sum_{k=0}^{1-2s}a_k \zb^k.
\]
Analogously, for the case $m=2$ we arrive at a similar conclusion. Solutions exist only for $\lambda=0$ and $s<\frac12$, implying
\[
  \tilde\eta =   z P^{s-1} \sum_{k=0}^{1-2s}a_k \zb^k.
\]
This proves the following
\begin{thm}
  Let $\eta$ have weights $w;p,q]$ with $2s=p-q\le1$ then~\eqref{eq:64} has solutions only for $\lambda=\pm\frac12$. For each $\lambda$ the eigenspace has dimension $2-2s$ over $\CC$. The solutions are given in terms of a complex polynomials $h(\zb)$ of degree at most $1-2s$ in the form
  \[
    \eta = \alpha^{(w+p-1)}A^{(p+q)/2} h(\zb)P^{s-1} \, z^{\lambda+\tfrac12}.
  \]
\end{thm}
In the case of interest $\eta$ has weights $[0;1,1]$ and is real valued. Since $s=0$ the eigenspaces are 1-dimensional, and since every translation is a linear combination of these eigenfunctions we can write it in the form
\[
  \eta = A \frac{(\bar{a}\zb + \bar{b})(az + b)}{P} \qquad \text{ for } a,b \in \CC.
\]

\subsection{The equation for the co-curvature}
\label{sec:equat-co-curv}

Next, we study the equation
\begin{equation}
\etbc \B = -\ethc\cQ,\label{eq:66}
\end{equation}
for a conformal density $\B$ of type $[-2;2,-2]$. It is conformally invariant, identical to the equation $\etp\B=-\ethc\cQ$, and can be evaluated in the standard gauge. Using the well-known properties of $\etp$ on the unit-sphere we find that when acting on spin-weight $2$ quantities such as $\B$ it is an isomorphism, so that \eqref{eq:66} has a unique solution.

We can in fact find this solution as follows. Consider
\[
  \begin{multlined}
    \ethc\cQ = \eth(K-\etp\tau) - 2 \tau (K - \etp\tau) = \eth K - \eth\etp \tau - 2 \tau K + 2\tau \etp\tau\\
    = \eth K - \etp\eth \tau + 2\tau \etp\tau = \eth K -\etp(\eth\tau - \tau^2).
\end{multlined}
\]
In the cylinder gauge, when $\eth K=0$ we have $\ethc\cQ = - \etp(\eth\tau - \tau^2)$ so that uniqueness of the solution implies that $\B = \eth\tau - \tau^2$ \emph{in that gauge}.

\subsection{The equation for the translations}
\label{sec:equat-transl}

The final topic in this appendix is the equation \eqref{eq:39}. We first write out the equation in terms of the usual $\eth$ operators. We obtain
\[
  \eth^2U + (\eth\tau - \tau^2) U = \B U.
\]
In the cylinder gauge, using the expression for $\B$ as found above this equation reduces simply to
\[
  \eth^2U = 0,
\]
and the properties of $\eth$ on the unit-sphere imply that $U$ is a linear combination of the lowest four spherical harmonics, i.e., in terms of stereographic coordinates it can be written as
\begin{equation}
  \label{eq:67}
  U = \frac{a_0+a_1z+a_2\zb+a_3z\zb}{1+z\zb}
\end{equation}
for arbitrary coefficients $a_i$. When restricted to real $U$ this implies that the solution space of~\eqref{eq:39} has real dimension 4.

\section{Stokes' theorem in cGHP formalism}
\label{sec:stokes-theorem-cghp}

Here we show briefly how Stokes' theorem applied to 2-forms on $\scri$ can be expressed in terms of the cGHP formalism. We start with a general 2-form on $\scri$
\[
  \mub = \mu_0 \i\mbbar\mbb + \mu_1 \i \lbb\mbb - \mu_2 \i \lbb\mbbar.
\]
A real-valued $\mub$ must satisfy $\bar\mu_0=\mu_0$, $\bar\mu_1=\mu_2$. Stokes' theorem then relates the integral of the exterior derivative $\dd\mub$ over a 3-dimensional submanifold of $\scri$ to the integral of $\mub$ over its boundary. Thus, one needs to compute $\dd\mub$. This is a straightforward calculation involving the exterior derivatives of the frame vectors expressed in terms of spin-coefficients (see e.g., \cite{Penrose:1984a}). If one assumes that $\mub$ is a scalar valued 2-form then the result can be expressed in terms of the GHP operators acting on the components
\[
  \dd\mub = \left[\left(\thorp \mu_0 - 2\rho'\mu_0\right) + \left(\eth\mu_1 - \tau \mu_1 \right) + \left(\etp\mu_2 - \taub \mu_2 \right)\right] \, \i \lbb\mbbar\mbb.
\]
With $\mub$ being scalar valued and assuming that it is conformally invariant, i.e., with conformal weight $w=0$ the components have well-defined weights $\mu_0:[-2;0,0]$, $\mu_1:[-1;-2,0]$ and $\mu_2:[-2;0,-2]$ so that the exterior derivative may be written in the form
\begin{equation}
  \label{eq:72}
  \dd\mub = \left[\thorpc \mu_0  + \ethc\mu_1 + \etbc\mu_2 \right] \, \i \lbb\mbbar\mbb.
\end{equation}

\printbibliography
%\listoftodos[Remarks in the margins]

% \bibliography{papers,bondimass}
%\bibliographystyle{cqg}
\end{document}